\documentclass{tlp}
\usepackage{aopmath}
\usepackage{bbold}

\def\ar{\leftarrow}
\def\rar{\rightarrow}
\def\lrar{\leftrightarrow}
\def\beq{\begin{equation}}
\def\eeq#1{\label{#1}\end{equation}}
\def\ba{\begin{array}}
\def\ea{\end{array}}

\def\G{\Gamma}
\def\D{\Delta}
\def\seq{\Rightarrow}
\def\r#1#2{\frac{\textstyle #1}{\textstyle #2}}
\newtheorem{remark}{Remark}
\newtheorem{example}{Example}

\submitted{4 December 2013}
\revised{21 February 2014}
\accepted{10 March 2014}

\title[Equivalence of Infinitary Formulas]
{\bf On Equivalence of Infinitary Formulas\\
under the Stable Model Semantics}

\author[Harrison, Lifschitz, and Truszczynski]{ Amelia Harrison, Vladimir
Lifschitz\\
University of Texas 
\and Miroslaw Truszczynski \\
University of Kentucky}

\begin{document}
\date{}
\maketitle

\begin{abstract}
Propositional formulas that are equivalent in intuitionistic logic, or in its
extension known as the logic of here-and-there, have the same stable models. 
We extend this theorem to propositional formulas
with infinitely long conjunctions and disjunctions and show how to apply this
generalization to proving properties of aggregates in answer set programming.
\end{abstract}


\section{Introduction} \label {sec:intro}

This note is about the extension of the stable model semantics to infinitary
propositional formulas defined by \citeANP{tru12} \citeyear{tru12}.  
That extension,
introduced originally as a tool for proving a theorem about the logic FO(ID),
has been used also to prove a new generalization of Fages' theorem
\cite{lif12a}.

One of the reasons why stable models of infinitary formulas are important is
that they are closely related to aggregates in answer set programming (ASP). The
semantics of aggregates proposed by \citeANP{fer05}
(\citeyearNP{fer05}, Section~4.1) treats a ground
aggregate as shorthand for a propositional formula.  An aggregate with
variables has to be grounded before that semantics can be applied to it.
For instance, to explain the precise meaning of the expression $1\{p(X)\}$
(``there exists at least one object with the property~$p$'') in the body of
an ASP rule we first rewrite it as
$$1\{p(t_1),\dots,p(t_n)\},$$
where $t_1,\dots,t_n$ are all ground terms in the language of the program,
and then turn it into the propositional formula
\beq
p(t_1)\lor\cdots\lor p(t_n).
\eeq{fd}
But this description of the meaning of $1\{p(X)\}$ implicitly assumes that
the Herbrand universe of the program is finite.  If the program contains
function symbols then an infinite disjunction has to be used instead
of~(\ref{fd}).

There is nothing exotic or noncomputable about
ASP programs containing both aggregates and function symbols.  For instance,
the program
$$\ba l
p(f(a))\\
q \ar 1\{p(X)\}\\
\ea$$
has simple intuitive meaning, and its stable model $\{p(f(a)),q\}$
can be computed by Version 3 of the answer set solver {\sc clingo}.\footnote{\tt http://potassco.sourceforge.net.} More generally, stable models
of infinitary propositional formulas in the sense of \citeANP{tru12} 
\citeyear{tru12} can be used to define the semantics of aggregates in 
the input language of {\sc clingo} \cite{har13b}; this is our main
motivation for studying their properties.
\begin{remark}
Attempts to define the semantics of aggregates for other ASP languages 
encounter similar difficulties if the Herbrand universe is infinite.  For
instance, the definition of a ground instance of a rule in Section~2.2 of the
ASP Core document ({\tt https://www.mat.unical.it /aspcomp2013/files/ASP-CORE-2.0.pdf}, Version 2.02)
talks about replacing the expression $\{e_1;\ldots;e_n\}$ in a rule with
a set denoted by $\hbox{inst}(\{e_1;\ldots;e_n\})$.  But that set can be infinite
and then it cannot be included in a rule.
\end{remark}

Our goal here is to develop methods for proving that pairs~$F$,~$G$ of
infinitary formulas have the same stable models.  From the
results of \citeANP{pea97} \citeyear{pea97} and \citeANP{fer05} \citeyear{fer05} 
we know that in the case of
grounded logic programs in the sense of \citeANP{gel88} \citeyear{gel88} and, 
more generally, sets of 
finite propositional formulas, it is sufficient to check that the equivalence
$F\lrar G$ is provable intuitionistically.  Some extensions of intuitionistic
propositional logic, including the logic of here-and-there, can be used as
well.  In this note, we extend these results to deductive systems of
infinitary propositional logic.

This goal is closely related to the idea of strong equivalence~
(Lifschitz, Pearce, Valverde, 2001)\nocite{lif01}.
The provability of $F\lrar G$ in the  deductive systems of
infinitary logic described below guarantees not only that~$F$ and~$G$
have the same stable models, but also that for any set $\mathcal H$ of
infinitary formulas, $\mathcal{H}  \cup \{F\}$ and~$\mathcal{H} \cup \{G\}$
have the same stable models.

We review the stable model semantics of infinitary propositional formulas in 
Section~\ref{sec:sm}. Then we define a basic infinitary system of natural deduction,
similar to propositional intuitionistic logic (Section~\ref{sec:basic}), and 
study its properties (Section~\ref{sec:properties}). The main theorem is stated
and proved in Section~\ref{sec:main}, and applied to examples involving aggregates 
in Section~\ref{sec:agg}. A useful extension of the basic system is discussed
in Section~\ref{sec:ext}.

A preliminary report on this work was presented at the 2013 
International 
Conference on Logic Programming and Nonmonotonic Reasoning 
(Harrison, Lifschitz, Truszczynski, 2013) \nocite{har13}.
 
\section{Stable Models of Infinitary Propositional Formulas} \label{sec:sm}
The definitions of infinitary formulas and their stable models given below are 
equivalent to the definitions proposed by \citeANP{tru12} \citeyear{tru12}.

Let $\sigma$ be a propositional signature,
that is, a set of propositional atoms.  The sets
$\mathcal{F}^\sigma_0$, $\mathcal{F}^\sigma_1$, $\ldots$ are defined as follows:
\begin{itemize}
\item $\mathcal{F}^\sigma_0=\sigma$,
\item $\mathcal{F}^\sigma_{i+1}$ is obtained from $\mathcal{F}^\sigma_{i}$ by 
adding expressions $\mathcal{H}^\lor$ and $\mathcal{H}^\land$ for all subsets 
$\mathcal{H}$ of $\mathcal{F}^\sigma_i$, and expressions $F\rar G$ for all 
$F,G\in\mathcal{F}^\sigma_i$.
\end{itemize}
The elements of $\bigcup^{\infty}_{i=0}\mathcal{F}^\sigma_i$ are called {\sl 
(infinitary) formulas} over $\sigma$.
\begin{remark}
This definition differs from the syntax introduced in early work on
infinitary propositional formulas \cite{sco58,kar64} in several ways.
It treats the collection $\mathcal{H}$ of conjunctive or disjunctive terms as
a set, rather than a family indexed by ordinals.  Thus there is no
order among conjunctive or disjunctive terms in this framework, and there can
be no repetitions among them.  More importantly, there is no restriction here
on the cardinality
of the set of conjunctive or disjunctive terms.  On the other hand, in the
hierarchy $\mathcal{F}^\sigma_i$ of sets of formulas, $i$ is a natural number;
transfinite levels are not allowed.
\end{remark} 

A set $\mathcal{H}$ of formulas is {\sl bounded} if it is 
contained in one of the sets $\mathcal{F}^\sigma_i$. For a bounded
set $\mathcal{H}$ of formulas, $\mathcal{H}^\land$ and $\mathcal{H}^\lor$ are
infinitary formulas.  

The symbol $\bot$ will be understood as an abbreviation for
$\emptyset^{\lor}$; $\neg F$ stands for $F\rar\bot$, and $F\lrar G$ 
stands for $(F\rar G)\land(G\rar F)$.

We will write $\{F,G\}^\land$ as $F\land G$, and
$\{F,G\}^\lor$ as $F\lor G$.  This convention allows us to view
finite propositional formulas over~$\sigma$
as a special case of infinitary formulas. For any bounded family 
$\{F_\alpha\}_{\alpha \in A}$ of formulas, we denote the formula 
$\{F_\alpha: {\alpha \in A}\}^\land$  by $\bigwedge_{\alpha \in A} F_\alpha$, 
and similarly for disjunctions. 

Subsets of a signature~$\sigma$ will be also called its {\sl interpretations}.
The satisfaction relation between an interpretation~$I$ and a formula~$F$ is 
defined as follows:
\begin{itemize}
\item For every $p\in \sigma$, $I\models p$ if $p\in I$.
\item $I\models\mathcal{H}^\land$ if for every formula $F\in\mathcal{H}$, $I\models F$.
\item $I\models\mathcal{H}^\lor$ if there is a formula $F\in\mathcal{H}$ such that $I\models F$.
\item $I\models F\rar G$ if $I\not\models F$ or $I\models G$.
\end{itemize}

We say that $I$ {\sl satisfies} a set $\mathcal{H}$ of formulas if $I$ 
satisfies all elements of $\mathcal{H}$. 
Two sets of  formulas are {\sl equivalent} to each other if they are
satisfied by the same interpretations.
A formula~$F$ is {\sl tautological} if it is satisfied by all
interpretations.

The {\sl reduct} $F^I$ of a formula~$F$ with respect to an
interpretation~$I$ is defined as follows:
\begin{itemize}
\item For $p\in \sigma$, $p^I=\bot$ if $I\not\models p$; otherwise $p^I=p$.
\item $(\mathcal{H}^\land)^I=\{G^I\ |\ G\in\mathcal{H}\}^\land$.
\item $(\mathcal{H}^\lor)^I=\{G^I\ |\ G\in\mathcal{H}\}^\lor$.
\item $(G\rar H)^I=\bot$ if $I\not\models G\rar H$; otherwise $(G\rar H)^I=G^I\rar H^I$.
\end{itemize}
The {\sl reduct} $\mathcal{H}^I$ of a set $\mathcal{H}$ of formulas is the set 
consisting of the reducts of the elements of $\mathcal{H}$. 
An interpretation~$I$ is a {\sl stable model} of a set $\mathcal{H}$ of  
formulas if it is minimal w.r.t. set inclusion among the interpretations satisfying 
$\mathcal{H}^I$; a stable model of a formula $F$ is a stable model of singleton 
$\{F\}$.  This is a straightforward extension of the 
definition of a stable model due to \citeANP{fer05} \citeyear{fer05}
to infinitary formulas.

It is easy to see that~$I\models F^I$ iff $I\models F$.  It follows that
every stable model of~$\mathcal{H}$ satisfies~$\mathcal{H}$. 

\section{Basic Infinitary System of Natural Deduction} \label{sec:basic}

Inference rules of the deductive system described below are similar to the
standard natural deduction rules of propositional logic.\footnote{See, for
instance, \cite[Section~1.2.1]{lif08b}.} Its
derivable objects are {\sl (infinitary) sequents\/}---expressions
of the form $\G\seq F$, where~$F$ is an infinitary formula, and $\G$ is a
{\sl finite} set of infinitary formulas (``$F$ under assumptions $\G$'').  To
simplify notation, we will write $\G$ as a list.
We will identify a sequent of the form $\seq F$ with the formula~$F$.

There is one axiom schema $F\seq F$.  The inference rules are the introduction
and elimination rules for the propositional connectives
$$\begin{array}{ll}
\!(\land I)\;\r	{\G\seq H \quad\hbox{for all }H\in\mathcal H}
		{\G\seq \mathcal H^\land}
&\quad		
(\land E)\;\r	{\G\seq \mathcal H^\land}
		{\G\seq H}
	\quad(H\in\mathcal H)

\\ \\
\!(\lor I)\;\r	{\G\seq H}
		{\G\seq \mathcal H^\lor}
\quad(H\in\mathcal H)
&\quad				
(\lor E)\;\r{\G\seq \mathcal H^\lor \qquad \D,H \seq F
    \quad\hbox{for all }H\in\mathcal H}
	    {\G,\D\seq F}
\\ \\
\!(\rar\!\! I)\;\r	{\G,F\seq G}
			{\G\seq F\rar G}
&\quad				
(\rar\!\! E)\;\r	{\G\seq F\quad \D\seq F \rar G}
			{\G,\D\seq G},
\end{array}$$
where $\mathcal{H}$ is a bounded set of formulas, 
and the weakening rule
$$\ba{l}
(W)\;\r {\G\seq F}{\G,\D\seq F}.\qquad\qquad\qquad
\ea$$
\begin{remark}
The usual conjunction introduction rule is 
$$\r{\G\seq F \quad \D\seq G}{\G,\D\seq F\land G};$$
the corresponding infinitary rule above is similar to the 
more restrictive version:
$$\r{\G\seq F \quad \G\seq G}{\G\seq F\land G}.$$ 
In the presence of the weakening rule $(W)$, the two versions are equivalent
to each other. The situation with disjunction elimination is similar. 
The usual contradiction rule 
$$\ba{l}
(C)\;\r {\G\seq \bot}{\G \seq F}\qquad\qquad\qquad
\ea$$
is a special case of ($\lor E$).  
We do not include the law of the excluded middle in the set
of axioms, so that this deductive system is similar to intuitionistic,
rather than classical, propositional logic. 
\end{remark}

The set of {\sl theorems of the basic system} is the smallest set of sequents
that includes the axioms of the system and is closed under the application
of its inference rules.
We say that  formulas~$F$ and~$G$ are {\sl equivalent in the basic
system} if $F\lrar G$ is a theorem of the basic system.  The reason why we are
interested in this relation is that formulas equivalent in the basic system
have the same stable models, as discussed in Section~\ref{sec:main} below.

\begin{example}
Let $\{F_i\}_{i \in \mathbb{N}}$ be a bounded family of formulas. 
We will check that the formula 
\beq
F_0\land\bigwedge_{i\geq 0}(F_i\rar F_{i+1})
\eeq{ex1}
is equivalent in the basic system to the formula
$\bigwedge_{i\geq 0}F_i$.  The sequent
$$F_0\land\bigwedge_{i\geq 0}(F_i\rar F_{i+1})
\seq\; F_0\land\bigwedge_{i\geq 0}(F_i\rar F_{i+1})$$
belongs to the set of theorems of the basic system.  Consequently so do
the sequents
$$F_0\land\bigwedge_{i\geq 0}(F_i\rar F_{i+1})
\;\seq\; F_0$$
and
$$F_0\land\bigwedge_{i\geq 0}(F_i\rar F_{i+1})
\seq\; F_j\rar F_{j+1}$$
for all $j\geq 0$.  Consequently the sequents
$$F_0\land\bigwedge_{i\geq 0}(F_i\rar F_{i+1})
\seq\; F_j$$
for all $j\geq 0$ belong to the set of theorems as well (by induction on~$j$).
Consequently so does the sequent
$$
F_0\land\bigwedge_{i\geq 0}(F_i\rar F_{i+1})\;\seq\;\bigwedge_{i\geq 0}F_i.
$$
A similar argument (except that induction is not needed) shows that the
sequent
$$
\bigwedge_{i\geq 0}F_i\;\seq\; F_0\land\bigwedge_{i\geq 0}(F_i\rar F_{i+1})
$$
is a theorem of the basic system also.  Consequently so is the sequent
$$
\seq\; F_0\land\bigwedge_{i\geq 0}(F_i\rar F_{i+1})\ \lrar\ 
\bigwedge_{i\geq 0}F_i.
$$
\end{example}

This argument could be expressed more concisely, without explicit
references to the set of theorems of the basic system, as follows.
Assume~(\ref{ex1}).  Then $F_0$ and, for every $i\geq 0$, $F_i\rar F_{i+1}$.
Then, by induction, $F_i$ for every~$i$.  And so forth.  This style of
presentation is used in the next example.

\begin{example}\label{e2}
Let $\{F_\alpha\}_{\alpha \in A}$ be a bounded family of formulas, and let 
$G$ be a formula.  Let us show that 
\beq
\left ( \bigvee_{\alpha \in A}F_\alpha\right ) \rar G
\eeq{ex2l}
is equivalent in the basic system to the formula
\beq
\bigwedge_{\alpha \in A}(F_\alpha\rar G).
\eeq{ex2r}
Left-to-right: assume~(\ref{ex2l}) and~$F_\alpha$.  Then $\bigvee_{\alpha \in A}F_\alpha$,
and consequently~$G$.  Thus we established~$F_\alpha\rar G$ under
assumption~(\ref{ex2l}) alone for every~$\alpha$, and consequently
established~(\ref{ex2r}) under this assumption as
well.  Right-to-left: assume~(\ref{ex2r})
and $\bigvee_{\alpha \in A}F_\alpha$, and consider the cases corresponding to the
disjunctive terms of this disjunction.  Assume~$F_\alpha$.  From~(\ref{ex2r}),
$F_\alpha\rar G$, and consequently~$G$.  Thus we established~$G$ in each case,
so that~(\ref{ex2l}) follows from~(\ref{ex2r}) alone.
\end{example}

\medskip

It is easy to see that the infinitary counterparts of the intuitionistically
provable De Morgan's laws 
\beq
\bigvee_{F \in \mathcal H}\neg F \rar \neg\bigwedge_{F \in \mathcal H} F 
\eeq{dem1}
and 
\beq
\bigwedge_{F \in \mathcal H}\neg F \lrar \neg\bigvee_{F \in \mathcal H} F,
\eeq{dem2}
where $\mathcal{H}$ is a bounded set of formulas, 
are theorems of the basic system. So are the infinitary distributivity 
laws
\beq
\left ( \bigvee_{\{F_i\}_{i \in I}}\ \  \bigwedge_{i \in I} F_i \right ) \rar
\left ( \bigwedge_{i \in I}\ \ \bigvee_{F \in \mathcal{H}_i} F \right )  
\eeq{d1}
and
\beq
\left ( \bigvee_{i \in I}\ \ \bigwedge_{F \in \mathcal{H}_i} F \right ) \rar 
\left ( \bigwedge_{\{F_i\}_{i \in I}}\ \  \bigvee_{i \in I} F_i \right ) 
\eeq{d2}
for every non-empty family $\{\mathcal{H}_i\}_{i \in I}$ of sets of formulas 
such that its union is bounded. The disjunction in the antecedent of 
(\ref{d1}) and the conjunction
in the consequent of (\ref{d2}) extend over all elements 
$\{F_i\}_{i \in I}$ of the Cartesian product of the family 
$\{\mathcal{H}_i\}_{i \in I}$. In Section \ref{sec:ext} we discuss an extension
of the basic system in which we postulate the converses of implications
(\ref{dem1}), (\ref{d1}), and (\ref{d2}).  

\section{Properties of the Basic System}\label{sec:properties}

The following property of the basic system is easy to verify.

\begin{proposition}\label{prop:int}
If a sequent consisting of finite formulas is intuitionistically provable then
it is a theorem of the basic system.
\end{proposition}

Recall that we define the set of theorems of the basic system to be the smallest
set of formulas that includes the axioms and is closed under the inference
rules. 
When we want to prove that every theorem of the basic system has a certain
property~$P$, it is clearly sufficient to check that every axiom has the
property~$P$, and that the set of sequents that have the property~$P$ is closed
under the application of the inference rules.
In this way we can establish, in particular, the following fact:

\begin{proposition}\label{prop:taut}
For any theorem $\Gamma\seq F$ of the basic system, the
formula $\G^\land\rar F$ is tautological.

\begin{remark}
The assertion of Proposition \ref{prop:taut} will remain
true even if we extend the set of axioms to include the law of the excluded
middle 
\beq
F \lor \neg F.
\eeq{lem} 
The converse is not true, however, even in the
presence of this axiom schema. This fact can be established by standard
methods used to prove incompleteness in infinitary logic, which utilize the
Downward L{\"o}wenheim-Skolem Theorem and the Mostowski 
Collapsing Lemma.\footnote{John Schlipf, personal communication.} We can make the 
system complete by postulating the following infinitary version of the law
of the excluded middle:
\beq
\bigvee_{J \subseteq I}\left (\bigwedge_{j \in 
J} F_j \land \bigwedge_{j \in I \setminus J}
\neg F_j \right ),
\eeq{ilem}
for any non-empty bounded family $\{F_i\}_{i \in I}$ of formulas.\footnote{ 
The proof of this fact is similar to the proof of completeness
of classical propositional logic due to \citeANP{kal35} \citeyear{kal35}. 
For any interpretation $I$, let $L_I$ denote the conjunction of the 
corresponding set of literals.  It is easy to check by induction
that for any formula $F$, 
$L_I \rar F$ is a theorem of the basic system if $I$ satisfies $F$, and 
$L_I \rar \neg F$ is a theorem of the basic system otherwise. 
The completeness of the basic system with (\ref{ilem}) added as an axiom 
schema easily follows.}
\end{remark}\label{ft}

\end{proposition}

Let $\sigma$ and $\sigma'$ be disjoint signatures.  In this section, a
{\sl substitution} is a bounded family of formulas over $\sigma$ with 
index set $\sigma'$. 
For any substitution~$\phi$
and any formula~$F$ over the signature $\sigma \cup \sigma'$, $\phi F$
stands for the formula over~$\sigma$ formed as follows:
\begin{itemize}

   \item If $F\in \sigma$ then $\phi F = F$. 
   \item If $F\in \sigma'$ then $\phi F = \phi_F$.  
   \item If $F$ is $\mathcal{H}^\land$ then $\phi F =
         \left \{\phi G\ |\ G \in \mathcal{H}\right \}^\land$. 
   \item If $F$ is $\mathcal{H}^\lor$ then $\phi F =
         \left \{\phi G\ |\ G \in \mathcal{H}\right \}^\lor$. 
   \item If $F$ is $G \rightarrow H$ then $\phi F =
         \phi G \rightarrow \phi H$. 
\end{itemize} 
Formulas of the form $\phi F$ will be called {\sl instances} of~$F$.

\begin{proposition}\label{prop:inst}
If~$F$ is a theorem of the basic system then every instance of~$F$ is a theorem
of the basic system also.
\end{proposition}

\begin{proof}
The notation $\phi F$ extends to sequents in a natural way.
The property ``$\phi S$ is a theorem of the basic system''
holds for every axiom~$S$ of the basic system, and it is preserved by all
inference rules.
\end{proof}

We will refer to Proposition~\ref{prop:inst} as the {\sl substitution property}
of the basic system.

\begin{example}\label{e3}
We will show that for any 
formulas~$F$,~$G$, the formula $\neg(F\lor G)$ is equivalent
to $\neg F \land \neg G$ in the basic system.  Note first that the formula
\beq
\neg(p\lor q) \lrar \neg p \land \neg q
\eeq{int}
is intuitionistically provable.  By Proposition~\ref{prop:int}, it follows that it
is a theorem of the basic system.  The equivalence
$$\neg(F\lor G) \lrar \neg F \land \neg G$$
is an instance of~(\ref{int}): take $\phi_p=F$, $\phi_q=G$.  By the
substitution property, it follows that it is a theorem of the basic
system as well.
\end{example}

\begin{proposition}\label{prop:repl}
For any substitutions~$\phi$, $\psi$ with the same index set, the implication 
$$\bigwedge_{p} (\phi_p \leftrightarrow \psi_p) 
\rightarrow (\phi F \leftrightarrow \psi F)$$
(where $p$ ranges over the indices) is a theorem of the basic system. 
\end{proposition}

\begin{proof}
The proof is by induction on~$j$ such that $F\in\mathcal{F}_j^{\sigma\cup\sigma'}$, and
it considers several cases, depending on the syntactic form of~$F$.  Assume, for
instance, that $F$ is $\mathcal{H}^\lor$.
Then
$$\phi F = \left \{\phi G\ |\ G \in \mathcal{H}\right \}^\lor,\quad
\psi F = \left \{\psi G\ |\ G \in \mathcal{H}\right \}^\lor.$$ By the induction 
hypothesis, for each $G$ in $\mathcal{H}$, the implication 
    \beq
          \bigwedge_{p} (\phi_p \leftrightarrow \psi_p) 
          \rightarrow (\phi G \leftrightarrow \psi G)
    \eeq{ih}
          is a theorem of the basic system. We need to show that 
    \begin{equation}
          \bigwedge_{p} (\phi_p \leftrightarrow \psi_p) 
          \rightarrow \left(\{\phi G\ |\ G\in\mathcal{H}\}^\lor 
          \leftrightarrow \{\psi G\ |\ G\in\mathcal{H}\}^\lor\right)
    \end{equation} 
          is a theorem of the basic system also.  Assume
\beq
\bigwedge_{p} (\phi_p \leftrightarrow \psi_p)
\eeq{e1}
and $\left \{\phi G\ |\ G \in \mathcal{H}\right \}^\lor$,
and consider the cases corresponding to the terms of this disjunction.
Assume $\phi G$. Then, by~(\ref{ih}) and~(\ref{e1}), $\psi G$.  We can conclude
$\left \{\psi G\ |\ G \in \mathcal{H}\right \}^\lor$, that is,~$\psi F$.
So we established the implication $\phi F \rightarrow \psi F$.  The implication
in the other
direction is proved in a similar way.
\end{proof}

\begin{corollary}
If for every index~$p$, $\phi_p$ is equivalent to $\psi_p$ in the basic
system then $\phi F$ is equivalent to $\psi F$ in the basic system.
\end{corollary}
We will refer to this corollary as the replacement property of
the basic system.

\begin{example}\label{e4}
The formula
\beq 
\bigwedge_{k \geq 1} (p_k \rar \neg p_k) \rar p_0
\eeq{fpsi} 
is equivalent to
\beq
\bigwedge_{k \geq 1} \neg p_k \,\rar\, p_0
\eeq{fphi} 
in the basic system, because~(\ref{fphi}) can be obtained from~(\ref{fpsi})
by replacing $p_k \rar \neg p_k$ with the intuitionistically equivalent
$\neg p_k$.  More formally, let $q_k\ (k \geq 1)$ be the indices
and let $F$ be $\bigwedge_{k \geq 1} q_k \,\rar\, p_0$.  For the substitutions
$$\phi_{q_k} = p_k\rar \neg p_k,\quad \psi_{q_k} = \neg p_k,$$
$\phi F$ is~(\ref{fpsi}), and $\psi F$ is~(\ref{fphi}).
By the replacement property,~(\ref{fpsi}) is equivalent 
to~(\ref{fphi}).  
\end{example}

\section{Relation of the Basic System to Stable Models}\label{sec:main}

{\sl Main Theorem}\\
\noindent For any set $\mathcal{H}$ of formulas, 
\begin{itemize}
  \item[(a)] if a formula $F$ is a theorem of the basic system then $\mathcal{H} 
  \cup \{F\}$ has the same stable models as $\mathcal{H}$;
  \item[(b)] if $F$ is equivalent to~$G$ in the basic system then~$\mathcal{H} 
  \cup \{F\}$ and~$\mathcal{H} \cup \{G\}$ have the same stable models.
\end{itemize}

\begin{lemma}  \label{lemma1}
For any formula $F$ and interpretation~$I$, if~$I$ does not satisfy~$F$ then
$F^I~\seq~\bot$ is a theorem of the basic system.
\end{lemma}

The proof is straightforward by induction on $i$ such that $F \in 
\mathcal{F}_i^\sigma$.

\medskip
By $\G^I$ we denote the set $\{G^I\ |\ G\in\G\}$; $(\G\seq F)^I$ stands for
$\G^I\seq F^I$.

\begin{lemma}  \label{lemma2}
For any sequent~$S$ and any interpretation~$I$, if~$S$ is a theorem of the basic
system then so is $S^I$.
\end{lemma}
\begin{proof}
Consider the property of sequents: ``$S^I$ is a theorem of the basic system.'' 
To prove the lemma, it suffices to show that all theorems of the basic system 
have that property. It is clear that the reduct of every axiom of the basic 
system is a theorem (of the basic system). Verifying that the set of sequents 
with that property is closed under inference rules follows the same pattern 
for all inference rules but those involving implication. Consider, for 
instance, disjunction elimination:
\beq
\r{\G\seq \mathcal H^\lor \qquad \D,H \seq F
    \quad\hbox{for all }H\in\mathcal H}
	    {\G,\D\seq F}
\eeq{de}
and assume that the reducts of all sequents that are premises of that rule 
are theorems. Because $(\mathcal H^\lor)^I$ is 
$(\mathcal H^I)^\lor$, all premises of the disjunction elimination rule:
$$\r{\G^I\seq (\mathcal H^I)^\lor \qquad \D^I,H^I \seq F^I
    \quad\hbox{for all }H\in\mathcal H}
	    {\G^I,\D^I\seq F^I}$$
are theorems. Therefore, so is the sequent $\G^I,\D^I\seq F^I$ and 
consequently, also the sequent $(\G,\D\seq F)^I$. 

Consider now the implication introduction rule:
$$\r	{\G,F\seq G}{\G\seq F\rar G}$$
and assume that the reduct $(\G,F\seq G)^I$
is a theorem. To show that $(\G\Rightarrow F\rightarrow G)^I$ is a theorem 
it suffices to show that $\G^I\Rightarrow(F\rightarrow G)^I$ is a theorem. 

\medskip\noindent
{\em Case 1:} $I$ satisfies $\G$. Since the sequent $(\G,F\seq G)^I$
is a theorem, so is the sequent $\G^I,F^I\seq G^I$. Thus, $\G^I\seq
F^I\rightarrow G^I$ is a theorem and so, $(\G^I)^\land\rightarrow 
(F^I\rightarrow G^I)$ is tautological. Since $I$ satisfies $\G$, the 
comment at the end of Section \ref{sec:sm} implies that $I$ 
satisfies $\G^I$. 
Consequently, $I$ satisfies $F^I\rightarrow G^I$ and, by the same comment
again, also $F\rightarrow G$. It follows that $(F\rar G)^I$ 
is $F^I\rar G^I$. Since the sequent $(\G,F\seq G)^I$ or, equivalently, 
the sequent $\G^I,F^I\seq G^I$ is a theorem, applying the rule
$$\r	{\G^I,F^I\seq G^I}{\G^I\seq F^I\rar G^I}$$
we obtain that $\G^I\Rightarrow F^I\rightarrow G^I$ is a theorem. Thus,
$\G^I\Rightarrow(F\rightarrow G)^I$, is a theorem, too.

\medskip\noindent
{\em Case~2:} $I$ does not
satisfy~$\G$.  Then $I$ does not satisfy one of the elements $H$ of $\G$.
By Lemma \ref{lemma1}, $H^I \seq \bot$ is a theorem, and 
$\G^I\seq (F\rar G)^I$ can be derived from $H^I \seq \bot$ by
rules~$(C)$ and~$(W)$. Thus, it is a theorem.

\medskip
Next, consider the implication elimination rule:
$$\r	{\G \seq F \qquad \D \seq F \rar G}{\G,\D \seq G}$$
and assume that the sequents $(\G \seq F)^I$ and $(\D \seq F \rar G)^I$ are 
theorems. 
We will show that
$(\G,\D\Rightarrow G)^I$ or, equivalently, $\G^I,\D^I\seq G^I$ is a theorem,
too.

\medskip\noindent
{\em Case 1:} $I$ satisfies $F \rar G$. Then $(F \rar G)^I$ is
$F^I \rar G^I$. Thus, the sequents $\G^I \seq F^I$ and $\D^I \seq F^I \rar 
G^I$ are theorems, and the claim follows by applying the rule 
$$\r	{\G^I \seq F^I \qquad \D^I \seq F^I \rar G^I}{\G^I,\D^I \seq G^I}.$$

\medskip\noindent
{\em Case~2:} $I$ does not
satisfy~$F \rar G$.  Then ~$(F \rar G)^I$ is~$\bot$ and so, 
$\D^I\seq\bot$ is a theorem. Moreover,
$\G^I,\D^I\seq G^I$ can be derived from $\D^I \seq \bot$ by
rules~$(C)$ and~$(W)$. Thus, $\G^I,\D^I\seq G^I$ is a theorem, too.
\end{proof}

\begin{proof}[Proof of Main Theorem]

(a)~Assume that~$F$ is
a theorem of the basic system.  By Lemma \ref{lemma2}, for any 
interpretation~$I$, 
$F^I$ is a theorem of the basic system, and consequently is tautological,
by Proposition~\ref{prop:taut}.
It follows that $\mathcal{H}^I$ and $(\mathcal{H}\cup {F})^I$ 
are satisfied by the same interpretations. 

(b)~Assume that~$F$ is equivalent to~$G$ in the basic system, that
is, $F \lrar G$ is a theorem of the basic system. By Lemma 
\ref{lemma2}, for every interpretation $I$, $(F \lrar G)^I$ is a 
theorem of the basic system. Moreover, by Proposition 
\ref{prop:taut}, $F \lrar G$ is tautological. Thus $(F \lrar G)^I
= F^I \lrar G^I$ and so, $F^I \lrar G^I$ is a theorem of the 
basic system. Consequently, $F^I \lrar G^I$ is tautological, that
is $F^I$ and $G^I$ are equivalent. 
It follows that $(\mathcal{H}\cup {F})^I$ and $(\mathcal{H}\cup {G})^I$ 
are satisfied by the same interpretations. 
\end{proof}

\section{Examples Involving Aggregates} \label{sec:agg}

As discussed in the introduction, infinitary formulas can be used
to precisely define the semantics of aggregates in ASP when the
Herbrand universe is infinite. In this section, we give two examples
demonstrating how the theory described above can be applied to prove 
equivalences between programs involving aggregates.

\begin{example}\label{e5}
Intuitively, the rule 
\beq
q(X) \ar 1\{p(X,Y)\}
\eeq{ag1}
has the same meaning as the rule 
\beq
q(X) \ar p(X,Y).
\eeq{r1}
To make this claim precise, consider first the result of grounding 
rule (\ref{ag1}) under the assumption that the Herbrand universe $C$
 is finite. 
In accordance with standard practice in ASP, we treat variable $X$ as global
and $Y$ as local. Then the result of grounding (\ref{ag1}) is the set of ground
rules
$$
q(a) \ar 1\{p(a,b)\ |\ b \in C\} 
$$
for all $a \in C$. In the spirit of the semantics
for aggregates proposed by \citeANP{fer05} [\citeyearNP{fer05}, Section~4.1]
these rules
have the same meaning as the propositional formulas
\beq
\left ( \bigvee_{b \in C} p(a, b) \right ) \rar q(a).
\eeq{s1}
Likewise, rule (\ref{r1}) can be viewed as shorthand for the set of formulas
\beq
p(a, b) \rar q(a) 
\eeq{s2}
for all $a,b \in C$. It easy to see that these sets of formulas are
intuitionistically equivalent.  
\end{example}

How can we lift the assumption that the Herbrand universe is finite? We can treat 
(\ref{s1}) as an infinitary formula, and show that the conjunction of formulas
(\ref{s1}) is equivalent to the conjunction of formulas (\ref{s2}) in the basic
system. The fact that  the conjunction of formulas (\ref{s2}) for all 
$b \in C$  is equivalent to (\ref{s1}) in the basic system follows 
from Example~\ref{e2} (Section \ref{sec:basic}). 

\begin{example}\label{e6}
Intuitively,
\beq
q(X) \ar 2\{p(X,Y)\}
\eeq{ag2}
has the same meaning as the rule 
\beq
q(X) \ar p(X,Y1),\ p(X, Y2),\ Y1 \not = Y2.
\eeq{r2}
To make this claim precise, consider the infinitary formulas corresponding to 
(\ref{ag2}): 
\beq
\left (\bigvee_{b \in C}p(a,b) \land \bigwedge_{b \in C} 
\left ( p(a, b) \rar \bigvee_{c \in C \atop c \not = b} p(a, c) 
\right ) \right ) \rar q(a)
\eeq{f0}
($a \in C$); see \cite[Section~4.1]{fer05} for details on representing 
aggregates with propositional formulas. The formulas corresponding
to~(\ref{r2}) are 
\beq
(p(a,b) \land p(a,c)) \rar q(a)  
\eeq{f1}
($a,b,c \in C,\ b \not = c$). We will show that the conjunction of 
formulas (\ref{f0}) is equivalent to the conjunction of formulas (\ref{f1}) in 
the basic system.

It is sufficient to check that for every $a \in C$, (\ref{f0}) is 
equivalent to the conjunction of formulas (\ref{f1}) over all $b,c \in 
C$ such that $b \neq c$. By Example \ref{e2}, this conjunction is 
intuitionistically equivalent to 
\beq
\left ( \bigvee_{b,c \in C \atop b \not = c} (p(a,b) \land p(a,c))
\right ) \rar q(a).
\eeq{f2}
By the replacement property of infinitary formulas, it suffices to check that 
the antecedents of (\ref{f0}) and (\ref{f2}) are equivalent to each other.

Left-to-right: assume 
\beq
\bigvee_{b \in C}p(a,b) \land \bigwedge_{b \in C} \left (
p(a, b) \rar \bigvee_{c \in C \atop c \not = b} p(a, c) \right ).  
\eeq{ante}
Then $\bigvee_{b \in C} p(a, b)$. We will
reason by cases, with one case corresponding to each possible value $b_0$
of $b$. Case $p(a, b_0)$: by the second conjunctive term of (\ref{ante}), 
$$
p(a, b_0) \rar  \bigvee_{c \in C \atop c \not = b_0} p(a,c).
$$
Then the consequent of this implication follows.
Again we will reason by cases, with one case for each value
$c_0$ of $c$ where $c_0 \neq b_0$. Case $p(a, c_0)$: then $p(a, b_0) \land 
p(a, c_0)$. Consequently 
\beq
\bigvee_{b,c \in C \atop b \not = c} p(a,b) \land p(a,c).
\eeq{ante2}

Right-to-left: assume (\ref{ante2}). 
We reason by cases, with one case for each pair $b_0,\ c_0$,
where $b_0 \neq c_0$. Case $p(a, b_0) \land p(a,c_0)$: from $p(a, b_0)$ 
we derive the first conjunctive term of (\ref{ante}); 
from $p(a, c_0)$ we derive $$\bigvee_{c \in C,
\atop c \neq b} p(a, c),$$ and consequently the implication  
$$p(a, b) \rar \bigvee_{c \in C \atop c \not = b} p(a, c).$$ 
The conjunction of these implications for all $b \in C$ is the second 
conjunctive term of~(\ref{ante}).
\end{example}
 
\section{The Extended System of Natural Deduction} \label{sec:ext}

In this section, we show that the assertion of the main theorem will remain
true if we extend the basic system by adding the axiom schema 
\beq
F \lor ( F \rar G ) \lor \neg G
\eeq{hosoi}
characterizing (in the finite case) the logic of here-and-there \cite{hos66}, and the 
converses to the implications discussed at the end of Section 
\ref{sec:basic}:
\beq
\neg\bigwedge_{F \in \mathcal H} F \rar \bigvee_{F \in \mathcal H}\neg F,  
\eeq{dem}
\beq
\left ( \bigwedge_{i \in I}\ \ \bigvee_{F \in \mathcal{H}_i} F \right ) \rar 
\left ( \bigvee_{\{F_i\}_{i \in I}}\ \  \bigwedge_{i \in I} F_i \right ),
\eeq{dist_cod}
and
\beq
\left ( \bigwedge_{\{F_i\}_{i \in I}}\ \  \bigvee_{i \in I} F_i \right ) \rar
\left ( \bigvee_{i \in I}\ \ \bigwedge_{F \in \mathcal{H}_i} F \right ).   
\eeq{dist_doc}
When all conjunctions and disjunctions are finite, formula (\ref{dem})
can be derived intuitionistically from (\ref{hosoi}), and (\ref{dist_cod})
and (\ref{dist_doc}) are intuitionistically provable. We do not know 
to what extent the additional axiom schemas postulated here are independent
when infinite conjunctions and disjunctions are allowed.   

In the extended system, we can derive the theorem 
\beq
\bigvee_{J \subseteq I}\left ( \neg \bigvee_{j \in I \setminus J} F_j \land 
\neg \neg \bigwedge_{j \in J} F_j \right )
\eeq{iwem}
for any non-empty bounded family $\{F_i\}_{i \in I}$ of formulas. (This is a 
generalization of the weak law of the excluded middle $ \neg F \lor \neg \neg F$
to sets of infinitary formulas, similar to the generalization of the law of the 
excluded middle given in Remark \ref{ft}. It is equivalent 
in the basic system to the special case of~(\ref{iwem}) corresponding to
a family with a single element.) Indeed 
$$
\bigwedge_{i \in I}\left (\neg F_i \lor \neg \neg F_i \right )
$$
is a theorem of the extended system because $\neg F_i \land \neg \neg F_i$ can
be intuitionistically derived from (\ref{hosoi})  with $F_i$ as $F$ and $\neg
F_i$ as $G$. Using (\ref{dist_cod}) 
we obtain 
$$
\bigvee_{J \subseteq I}\left ( \bigwedge_{j \in I \setminus J}
\neg F_j \land \bigwedge_{j \in J} \neg \neg F_j \right ),
$$
and (\ref{iwem}) follows by De Morgan's laws. 

In the extended system, we can also derive the theorem
\beq
\left ( F \rar \bigvee_{i \in I} G_i \right ) \rar \bigvee_{i \in I}
\left ( F \rar G_i \right ) 
\eeq{dist_ioo}
for any formula $F$ and non-empty family $\{G_i\}_{i \in I}$ of formulas. 
We use instantiations of (\ref{hosoi}) for all $G_i$ to obtain
$$
\bigwedge_{i \in I} F \lor (F \rar G_i) \lor \neg G_i. 
$$
By (\ref{dist_cod}) we obtain
\beq
\bigvee_{\{F_i\}_{i \in I}}\,\bigwedge_{i \in I} F_i 
\eeq{dist_ioo2}
where the disjunction extends over all elements 
$\{F_i\}_{i \in I}$ of the Cartesian product of the family $\{F, F \rar G_i, 
\neg G_i\}_{i \in I}$. We reason by cases, with one case corresponding to 
each disjunctive term $\bigwedge_{i \in I} F_i$ of (\ref{dist_ioo2}). 
If at least one of the formulas $F_i$ is $F$ then from the antecedent of 
(\ref{dist_ioo}) we can derive $\bigvee_{i \in I} G_i$, and the consequent 
of (\ref{dist_ioo}) immediately follows. If at least one of the formulas
$F_i$ is $F \rar G_i$ then the consequent of (\ref{dist_ioo}) is immediate
as well.  Otherwise, $\bigwedge_{i \in I} F_i$ is $\bigwedge_{i \in I} \neg G_i$. 
Then from the antecedent of (\ref{dist_ioo}) we can derive $\neg F$ 
and every disjunctive term of the consequent follows. 

It is easy to check that the properties of the basic system proved in 
Section \ref{sec:properties} hold for the extended system as well.

To show that the assertion of the main theorem applies to the extended
system we will prove the modification of Lemma \ref{lemma2} stated below.
The {\sl classical extended system} is obtained from the extended system by
replacing the axiom schema~(\ref{hosoi}) with the law of the excluded
middle~(\ref{lem}).

\begin{lemma}  \label{lemma3}
For any sequent~$S$ and any interpretation~$I$, if~$S$ is a theorem of the
extended system then $S^I$ is a theorem of the classical extended system.
\end{lemma}

\begin{proof}
It suffices to show that every theorem $S$ of the
extended has this property: ``$S^I$ is a theorem of the classical extended
system.'' We only need to check that the reducts of the axioms 
(\ref{hosoi})--(\ref{dist_doc}) have this property; the fact that 
the set of sequents with that property is closed under the inference rules
is checked in the same way as in the proof of Lemma \ref{lemma2}. 

Let $S$ be (\ref{hosoi}). Then $S^I$ is 
$$F^I \lor (F \rar G)^I \lor (\neg G)^I.$$ If $I \models G$ then the 
second disjunctive term is $F^I \rar G^I$, and the disjunction can be 
derived from $F^I \lor \neg F^I$. If $I \not \models G$ then the third 
disjunctive term is equivalent to $\neg \bot$.    

Let $S$ be (\ref{dem}). Since $S$ is tautological, $S^I$ is  
$$
\left ( \neg \bigwedge_{F \in \mathcal H} F \right )^I \rar \bigvee_{F \in 
\mathcal H} (\neg F)^I.  
$$
If $I$ satisfies the conjunction in the 
antecedent, then the antecedent is $\bot$. Otherwise, at least one 
disjunctive term in the consequent is equivalent to~$\neg \bot$.  
 
Let $S$ be (\ref{dist_cod}). Since $S$ is tautological, $S^I$ is
$$
\left ( \bigwedge_{i \in I}\ \ \bigvee_{F \in \mathcal{H}_i} F^I \right ) \rar 
\left ( \bigvee_{\{F_i\}_{i \in I}}\ \  \bigwedge_{i \in I} F_i^I \right ), 
$$
which is an instantiation of the same axiom schema. The reasoning for sequents of 
the form (\ref{dist_doc}) is similar. 
\end{proof} 

\noindent{\sl Main Theorem for the Extended System}\\
\noindent For any set $\mathcal{H}$ of formulas, 
\begin{itemize}
  \item[(a)] if a formula $F$ is a theorem of the extended system then $\mathcal{H} 
  \cup \{F\}$ has the same stable models as $\mathcal{H}$;
  \item[(b)] if $F$ is equivalent to~$G$ in the extended system then~$\mathcal{H} 
  \cup \{F\}$ and~$\mathcal{H} \cup \{G\}$ have the same stable models.
\end{itemize}

This assertion is derived from Lemma \ref{lemma3} in the same way that
the Main Theorem was derived from Lemma \ref{lemma2}, using the fact that
all theorems of the classical extended system are tautological. 

\begin{example}\label{e7}
Intuitively, the cardinality constraint 
$\{p(X)\}0$ 
(``the set of true atoms with form $p(X)$ has cardinality at most 0'') has 
the same meaning as the conditional literal 
$\bot : p(X)$ (``for all $X$, $p(X)$ is false'').  If we represent this
conditional literal by the infinitary formula
\beq
\bigwedge_{a \in C} \neg p(a) 
\eeq{f8r}
then this claim can be made precise by showing that~(\ref{f8r}) is
equivalent in the extended system to
the infinitary formula corresponding to $\{p(X)\}0$ in the sense of
\cite{fer05}:
\beq
\bigwedge_{A \subseteq C \atop A \not = \emptyset} \left ( \bigwedge_{a \in A} 
p(a)  \rar  \bigvee_{a \in C \setminus A} p(a) \right )
\eeq{f8l}
(where $C$ is the Herbrand universe).

It is easy to derive (\ref{f8l}) from (\ref{f8r}) in the basic system. 
The derivation of~(\ref{f8r}) from~(\ref{f8l}) will use the following 
instance of (\ref{iwem}): 
\beq
\bigvee_{A \subseteq C}\left (\neg \bigvee_{a \in C \setminus A} p(a) \land \neg 
\neg \bigwedge_{a \in A} p(a) \right ).
\eeq{iwemi}
We will reason by cases, with one case 
corresponding to each disjunctive term~$D_A$ in~(\ref{iwemi}).  
In the case that $A$ is 
empty,~(\ref{f8r}) follows from the first conjunctive term of~$D_A$ by 
De Morgan's law.
Otherwise,  assume $ \bigwedge_{a \in A} p(a)$. 
Then by~(\ref{f8l}), $\bigvee_{a \in C \setminus A} p(a)$, 
which contradicts the first conjunctive term of $D_A$. We
conclude $\neg \bigwedge_{a \in A} p(a)$, 
which contradicts the second conjunctive term of $D_A$.
So the assumptions~$D_A$ and~(\ref{f8l}) are 
contradictory. Consequently, they imply~(\ref{f8r}).
\end{example}

\section{Conclusion}

Two finite propositional formulas are strongly equivalent if and only if they
are equivalent in the logic of here-and-there \cite[Proposition~2]{fer05}.
The results of this note are similar to the if part of that theorem. 
We don't know how to extend the only if part to infinitary formulas.  It 
is not 
clear whether or not any axioms or inference rules not included in the 
extended system will be required. However, as we illustrated with
several examples, the results in this paper allow us to verify 
the equivalence of formulas involving aggregates. 

\section*{Acknowledgements}

Thanks to Fangkai Yang and to the anonymous referees for comments.

\bibliographystyle{acmtrans}
\bibliography{bib}

\end{document}